\documentclass{article}
\usepackage{spconf}
\ninept
\usepackage{graphicx}
\usepackage[absolute,showboxes]{textpos}
\TPMargin{5pt}
\newcommand{\copyrightstatement}{
\begin{textblock}{0.8}(0.1,0.92)    
\noindent
\footnotesize
\copyright 2020 IEEE.  Personal use of this material is permitted.  Permission from IEEE must be obtained for all other uses, in any current or future media, including reprinting/republishing this material for advertising or promotional purposes, creating new collective works, for resale or redistribution to servers or lists, or reuse of any copyrighted component of this work in other works.
\end{textblock}
}
\usepackage{amsmath, amssymb, amsfonts, amsthm, mathtools}
\usepackage{bm, cite, color}
\usepackage{algorithm, algorithmic}
\usepackage{multirow}
\theoremstyle{definition}
\newtheorem{theorem}{Theorem}

\newtheorem{prop}[theorem]{Proposition}

\DeclareMathOperator*{\htop}{\mathsf{H}}
\DeclareMathOperator*{\trace}{tr}

\DeclareMathOperator*{\image}{Im}

\DeclareMathOperator*{\minimize}{minimize}

\title{Overdetermined Independent Vector Analysis}
\name{Rintaro Ikeshita, Tomohiro Nakatani, Shoko Araki}
\address{NTT Communication Science Laboratories, NTT Corporation, Kyoto, Japan}

\begin{document}
\copyrightstatement
\maketitle

\begin{abstract}
We address the convolutive blind source separation problem for the (over-)determined case where
(i) the number of nonstationary target-sources $K$ is less than that of microphones $M$, and
(ii) there are up to $M - K$ stationary Gaussian noises that need not to be extracted.
Independent vector analysis (IVA) can solve the problem by separating into $M$ sources and selecting the top $K$ highly nonstationary signals among them,
but this approach suffers from a waste of computation especially when $K \ll M$.
Channel reductions in preprocessing of IVA by, e.g., principle component analysis have the risk of removing the target signals.
We here extend IVA to resolve these issues.
One such extension has been attained by assuming the orthogonality constraint (OC) that the sample correlation between the target and noise signals is to be zero.
The proposed IVA, on the other hand, does not rely on OC and exploits only the independence between sources and the stationarity of the noises.
This enables us to develop several efficient algorithms based on block coordinate descent methods with a problem specific acceleration.
We clarify that one such algorithm exactly coincides with the conventional IVA with OC, and also explain that the other newly developed algorithms are faster than it.
Experimental results show the improved computational load of the new algorithms compared to the conventional methods.
In particular, a new algorithm specialized for $K = 1$ outperforms the others.
\end{abstract}

\section{Introduction}
\label{sec:intro}

Blind source separation (BSS) is a problem of estimating the source signals from their observed mixtures~\cite{comon1994ica,cardoso1998bss}.
In this paper, we focus on the (over-)determined BSS where the number of nonstationary signals $K$ is less than that of microphones $M$, i.e., $K < M$.
There can be up to $M - K$ stationary Gaussian noises as long as the problem remains (over-)determined.\footnote{
	The assumption that there are at most $M - K$ noises is for the sake of developing efficient algorithms rigorously and can be violated to some extent (see Section~\ref{sec:exp}).
}
The goal is to extract $K$ nonstationary signals efficiently.
We do not care about separating the noises.

If the mixture is convolutive, \textit{independent vector analysis} (IVA~\cite{kim2007iva,hiroe2006iva}) is one of the most fundamental methods to solve BSS.
As a straightforward way, we can apply IVA as if there are $M$ nonstationary sources and select the top $K (< M)$ highly nonstaionary signals among the $M$ separated signals.
This method is, however, computationally intensive and does not run in realtime if $M$ is large.
To improve the computational efficiency, in preprocessing of IVA, we can reduce the number of channels up to $K$ by using principle component analysis or picking up $K$ channels with high SNR.
These channel reductions, however, have the risk of removing the target signals and often degrade the separation performance~\cite{scheibler2019over-iva}.

BSS methods for efficiently extracting just one or several target signals with specific properties such as non-Gaussianity has already been studied~\cite{friedman1974projection,huber1985projection,hyvarinen1997fastica,hyvarinen1999fastica,wei2015fastica,cruces2004blind}.
Among them, we only focus on methods in which the demixing filters can be optimized by using an \textit{iterative projection (IP)}~\cite{ono2010auxica,ono2011iva,ono2012iva-stereo} technique.
IP is a class of \textit{block coordinate descent (BCD)} methods specialized for optimizing the maximum-likelihood-based IVA, with the advantages of high computational efficiency and no hyperparameters.
Recently proposed \textit{OverIVA}~\cite{scheibler2019over-iva} (IVA for the overdetermined case), which is a multi-target-source extension of independent component/vector extraction (ICE/IVE~\cite{koldovsky2017ive-eusipco,koldovsky2018ive}), takes the advantages of IP, resulting in significantly improved computational cost of IVA.
In addition to IP, OverIVA as well as ICE/IVE also relies on the orthogonality constraint (OC) that the sample correlation between the separated target and noise signals is to be zero.
As OC is a heuristic assumption for developing efficient algorithms, the open problems here are (i) to theoretically clarify the validity of OC, and (ii) to further accelerate OverIVA.

In this paper, we propose a different approach for BSS, which we call OverIVA as well.
It does not employ OC and exploits only the independence between sources and the stationarity of the Gaussian noises (see Section~\ref{sec:model}).
For this model, we develop several efficient optimization algorithms based on BCD, all of which can be viewed as extensions of IP (see Section~\ref{sec:alg}).
We then prove that one such algorithm exactly coincides with OverIVA with OC~\cite{scheibler2019over-iva}, meaning that the stationarity of the Gaussian noises implicitly implies OC (see Section~\ref{sec:relation}).
We also describe that the other newly developed algorithms are expected to be faster than OverIVA with OC (see Section~\ref{sec:relation}), and confirm their effectiveness in the experiment (see Section~\ref{sec:exp}).

\section{Problem Formulation}
\label{sec:problem-formulation}

Suppose that $K$ nonstationary source signals are mixed with a stationary noise, or possibly a silent signal, of dimension $M - K$ and observed with $M$ microphones, where $1 \leq K \leq M - 1$.
In the time-frequency domain, the mixture signal, $\bm{x}(f,t)$, is modeled by
\begin{align} \label{eq:mixture}
& \hspace{1mm} \bm{x}(f,t) = A_s(f) \bm{s}(f,t) + A_z(f) \bm{z}(f,t) \in \mathbb{C}^{M},
\\
\nonumber
A_s(f) &= [\, \bm{a}_{1}(f), \ldots, \bm{a}_{K}(f) \,] \in \mathbb{C}^{M \times K}, \quad A_z(f) \in \mathbb{C}^{M \times (M - K)},
\\
\nonumber
\bm{s}(f,t) &= [\, s_{1}(f,t), \ldots, s_{K}(f,t) \,]^\top \in \mathbb{C}^K, \quad \bm{z}(f,t) \in \mathbb{C}^{M - K},
\end{align}
where $f \in \{1, \ldots,F \}$ and $t \in \{1, \ldots, T \}$ denote the frequency bin and time-frame indexes, respectively,
$\empty^{\top}$ is the transpose,
$\bm{a}_k(f)$ and $s_k(f,t)$ are the transfer function and signal for the target source $k$, respectively,
and $A_z(f)$ and $\bm{z}(f,t)$ are those for the noise source.

The  demixing matrix $W(f) = [ \bm{w}_1(f), \ldots, \bm{w}_M(f) ] \in \mathbb{C}^{M \times M}$ satisfying $W(f)^{\htop} [A_s(f), A_z(f)] = I_M$ translates \eqref{eq:mixture} into 
\begin{align}
\label{eq:s}
s_k(f,t) &= \bm{w}_k(f)^{\htop} \bm{x}(f,t) \in \mathbb{C}, \quad k \in \{ 1,\ldots,K \},
\\
\label{eq:z}
\bm{z} (f,t) &= W_z(f)^{\htop} \bm{x}(f,t) \in \mathbb{C}^{M - K},
\\
\label{eq:Wz}
W_z(f) &= [\, \bm{w}_{K + 1}(f), \ldots, \bm{w}_M(f) \,] \in \mathbb{C}^{M \times (M - K)},
\end{align} 
where $I_M$ is the identity matrix and $\empty^{\htop}$ is the Hermitian transpose.

The BSS problem dealt with in this paper is to recover the spatial images of the target nonstationary sources, $\{ \bm{a}_k(f) s_k(f,t) \}_{k,f,t}$
under the assumption that $K$ is given and the source signals are independent of each other.
Once the demixing matrices $\{ W_z(f) \}_f$ are obtained, we can estimate the spatial images using projection back technique~\cite{murata2001projection-back} as follows (here $\bm{e}_k \in \mathbb{C}^M$ denotes the unit vector whose $k$th element is equal to one and the others zero):
\begin{align}
\label{eq:projection-back}
\bm{a}_k(f) s_k(f,t) = (W(f)^{-\htop} \bm{e}_k) (\bm{w}_k(f)^{\htop} \bm{x}(f,t)) \in \mathbb{C}^M.
\end{align}

\section{Probabilistic model}
\label{sec:model}

We present the probabilistic model of the proposed OverIVA, which is almost identical to that of the ordinary IVA~\cite{kim2007iva,hiroe2006iva}.
The only difference is in the model of the stationary noise that we do not need to estimate.
In fact, the proposed model is defined by~\eqref{eq:s}--\eqref{eq:Wz} and the following~\eqref{eq:indep}--\eqref{eq:z-gauss}:
$\bm{s}_k(t) \coloneqq [\, s_k(1,t),\ldots, s_k(F,t) \,]^\top \in \mathbb{C}^F$, and
\begin{align}
\label{eq:indep}
& \hspace{-1.5cm} {\textstyle p( \{ \bm{s}_k(t), \bm{z}(f,t) \}_{k,f,t} ) = \prod_{k,t} p(\bm{s}_k(t)) \cdot \prod_{f,t}p(\bm{z}(f,t)), }
\\
\bm{s}_k (t) &\sim \mathbb{C} \mathcal{N} \left( \bm{0}, \lambda_k(t) I_F \right), \quad k \in \{1,\ldots,K \},
\\
\label{eq:z-gauss}
\bm{z}(f,t) &\sim \mathbb{C}\mathcal{N} \left( \bm{0}, I_{M - K} \right),
\end{align}
where $\{ \lambda_k(t) \}_t$ are frequency-independent time-varying variances modeling the power spectrum for source $k$.
The parameters to be optimized in the model are the demixing matrix $\bm{W} \coloneqq \{ W(f) \}_{f}$ and power spectra $\bm{\lambda} \coloneqq \{ \lambda_k(t) \}_{k,t}$.

Note that the noise model \eqref{eq:z-gauss} with the constant covariance matrix does not sacrifice generality.
At first glance, it seems better to employ $\bm{z}(f,t) \sim \mathbb{C}\mathcal{N} \left( \bm{0}, R(f) \right)$ with a general  $R(f)$ being a parameter to be optimized.
However, as we are not interested in the noise components, we can freely change the variables to satisfy \eqref{eq:z-gauss} using the ambiguity between $A_z(f)$ and $\{ \bm{z}(f,t) \}_{t = 1}^T$ as follows:
\begin{align}
\nonumber
A_z(f) \bm{z}(f,t) = (A_z(f) R(f)^{\frac{1}{2}}) (R(f)^{- \frac{1}{2}} \bm{z}(f,t)).
\end{align}

\section{Optimization}
\label{sec:alg}

We develop an algorithm for the maximum likelihood (ML) estimation of  the parameters $\bm{W}$ and $\bm{\lambda}$.
The ML estimation is attained by minimizing the negative log-likelihood $J$, which is computed as
\begin{align*}
J =& {\textstyle \sum_{k,t} \left[ \frac{ \| \bm{s}_k(t) \|^2 }{ \lambda_k(t) } + F \log \lambda_k(t) \right] + \sum_{f,t} \| \bm{z}(f,t) \|^2 }
\\
& \textstyle{  - 2 T \sum_f \log | \det W(f) | + C },
\end{align*}
where $C$ is a constant independent of the parameters.

As is often the case with IVA, the proposed algorithm updates $\bm{W}$ and $\bm{\lambda}$ alternately.
When $\bm{W}$ is kept fixed, $\bm{\lambda}$ can be globally optimized by $\lambda_k(t) = \frac{1}{F} \| \bm{s}_k(t) \|^2$.

In what follows, we will develop several computationally efficient algorithms that optimize $\bm{W}$, keeping $\bm{\lambda}$ fixed.
The objective function $J$ with respect to $\bm{W} = \{ W(f) \}_{f = 1}^F$ is additively separable for each frequency bin $f$.
This enables us to split the problem into $F$ independent problems, each of which is described as \eqref{problem-W} below. We hereafter abbreviate the frequency bin index $f$ to simplify the notation without mentioning it.
\begin{align} \label{problem-W}
& \minimize_{W}  J_{W}, \tag{P1}
\\
\nonumber
J_W &= {\textstyle \sum_{k = 1}^K \bm{w}_k^{\htop} G_{k} \bm{w}_k + \trace (W_z^{\htop} G_z W_z)  - 2 \log |\det W| },
\\
\nonumber
G_k &= {\textstyle \frac{1}{T} \sum_{t = 1}^T \frac{\bm{x}(t) \bm{x}(t)^{\htop}}{\lambda_{k}(t)} \in \mathbb{C}^{M \times M}, \quad k \in \{1,\ldots,K \} },
\\
\nonumber
G_z &= {\textstyle \frac{1}{T} \sum_{t = 1}^T \bm{x}(t) \bm{x}(t)^{\htop} \in \mathbb{C}^{M \times M} }.
\end{align}

\subsection{The algorithms to optimize $\bm{W}$}
\label{sec:opt-W}

If $K = M$ and the noise component does not exist, the problem~\eqref{problem-W} is known as the ML based ICA~\cite{pham1996ica-ml,pham2001ica-nonstationary,yeredor2012SeDJoCo,yeredor2009HEAD}.
Block coordinate descent (BCD) methods have been proposed to solve it~\cite{ono2010auxica,ono2011iva,ono2012iva-stereo,degerine2004ica-ML,degerine2006maxdet} and reported to be faster and to give higher separation performance than other algorithms such as the natural gradient method~\cite{amari1996natural-gradient} and FastICA~\cite{hyvarinen1997fastica,hyvarinen1999fastica} (see, e.g.,~\cite{ono2010auxica}).
The family of these BCD algorithms specialized to solve the ML based ICA is currently called an \textit{iterative projection (IP) method}~\cite{ono2018ip-ASJ}.
Even when $K < M$, IP can directly be applied to the problem~\eqref{problem-W} by translating the second term of $J_W$ as $\trace(W_z^{\htop} G_z W_z) = \sum_{k = K + 1}^M \bm{w}_k^{\htop} G_z \bm{w}_k$.
The problem is, however, a huge wasted computational cost especially in the case of $K \ll M$.

We therefore propose two different computationally efficient algorithms based on BCD to solve the problem \eqref{problem-W}.
We call them IP-1 and IP-2, respectively, because they can be viewed as extensions of the conventional IP, named IP-0 in this paper.
All algorithms exactly optimize one or several columns of $W$ in each iteration while keeping all other columns fixed.
The procedures of IP-1 and IP-2 as well as IP-0,  are summarized in Table~\ref{table:ip-procedure}.
As we will clarify in Section~\ref{sec:relation}, that the conventional OverIVA~\cite{scheibler2019over-iva} can be obtained from the proposed OverIVA by selecting IP-3 in Table~\ref{table:ip-procedure} as an optimization procedure of BCD.
To make the algorithms work properly, we need to introduce two technical but essential conditions (C1) and (C2):\footnote{
If (C1) is violated, the problem \eqref{problem-W} has no optimal solutions and the algorithms diverge to infinity (see {\cite[Proposition 1]{ike2019ilrma-ft}} for the proof).
In practice, we can always guarantee (C1) by adding small $\varepsilon I_M$ heuristically (see Algorithm~\ref{alg:iva}).
The condition (C2) is satisfied automatically if we initialize $W$ as nonsingular.
Intuitively, a singular $W$ implies $- \log | \det W | = +\infty$, which will never occur during optimization.
}
\begin{description}
\setlength{\itemsep}{-0pt}
\item[(C1)] $G_1, \ldots, G_K, G_z$ are positive definite.
\item[(C2)] Estimates of $W$ are always nonsingular during optimization.
\end{description}
As we will see, the algorithms are developed by exploiting the first order necessary optimality (stationary) conditions of the problem \eqref{problem-W} with respect to $\bm{w}_k$ $(k = 1,\ldots,K)$ and $W_z$,
which are expressed as follows (see, e.g.,~\cite{yeredor2009HEAD,yeredor2012SeDJoCo}):
%
\begin{alignat}{2}
\label{eq:KKT-w}
\frac{\partial J_W}{\partial \bm{w}_k^\ast} &= \bm{0}_M \quad &\Longleftrightarrow \quad W^{\htop} G_k \bm{w}_k &= \bm{e}_k \in \mathbb{C}^{M},
\\
\label{eq:KKT-Wz}
\frac{\partial J_W}{\partial W^\ast_z} &= O &\Longleftrightarrow \quad W^{\htop} G_z W_z &= E_z \in \mathbb{C}^{M \times (M - K)},
\end{alignat}
%
where 
$\empty^\ast$ denotes the element-wise conjugate,
$\bm{0}_M \in \mathbb{C}^M$ means the zero vector while $O \in \mathbb{C}^{M \times (M - K)}$ denotes the zero matrix,
and we define $E_z \coloneqq [\, \bm{e}_{K + 1}, \ldots, \bm{e}_{M} \,] \in \mathbb{C}^{M \times (M - K)}$.

\begin{table}
\begin{center}
{\footnotesize
\caption{Optimization procedure for each method}
\label{table:ip-procedure}
\vspace{-2mm}
\begin{tabular}{cl} \hline
Method & \multicolumn{1}{c}{Procedure}
\\ \hline
IP-0 (\S \ref{sec:ip-0}~\cite{ono2011iva}) & Optimize $\bm{w}_1, \ldots, \bm{w}_K, \ldots, \bm{w}_M$ cyclically.
\\
IP-1 (\S \ref{sec:ip-1}) & Optimize $\bm{w}_1, \ldots, \bm{w}_K, W_z$ cyclically.
\\
IP-2 (\S \ref{sec:ip-2}) & Optimize $\bm{w}_1$ and $W_z$ simultaneously (when $K = 1$).
\\
IP-3 (\S \ref{sec:relation}~\cite{scheibler2019over-iva}) & Optimize $\bm{w}_1, W_z, \bm{w}_2, W_z, \ldots, \bm{w}_K, W_z$ cyclically.
\\ \hline
\end{tabular}
}
\end{center}
\vspace{-8mm}
\end{table}

\vspace{-2mm}
\subsubsection{The review of the conventional IP: IP-0}
\label{sec:ip-0}

We review IP-0 for solving the problem~\eqref{problem-W}~(see \cite{ono2011iva} for the details).
In each iteration, IP-0 updates $\bm{w}_k$ for some $k = 1,\ldots,M~(> K)$ so that it globally minimizes $J_W$ with respect to $\bm{w}_k$ while keeping all other variables fixed.
This is achieved under (C1)--(C2) by
\begin{align}
\label{eq:ip0-u}
\bm{u}_k &\leftarrow (W^{\htop} G_k)^{-1} \bm{e}_k \in \mathbb{C}^M,
\\
\label{eq:ip0-w}
\bm{w}_k &\leftarrow \bm{u}_k \left( \bm{u}_k^{\htop} G_k \bm{u}_k \right)^{-\frac{1}{2}} \in \mathbb{C}^M.
\end{align}
Here, $G_k \coloneqq G_z$ $(k = K + 1, \ldots, M)$, and the $k$th column of $W$ in \eqref{eq:ip0-u}, $\bm{w}_k$, is set to the current estimate before update.
It is guaranteed that updating $\bm{w}_k$ with \eqref{eq:ip0-u}--\eqref{eq:ip0-w} will not increase the value of $J_W$.

\vspace{-0pt}
\subsubsection{The derivation of IP-1}
\label{sec:ip-1}

We derive IP-1 that updates $\bm{w}_1, \ldots,\bm{w}_K, W_z$ cyclically.
The update rules for $\bm{w}_1, \ldots, \bm{w}_K$ are given by \eqref{eq:ip0-u}--\eqref{eq:ip0-w} while that for $W_z$ is defined by \eqref{eq:ip1-Wz-fast} below.
To begin with, we enumerate all global solutions for the problem of minimizing $J_W$ with respect to $W_z$ when keeping $\bm{w}_1, \ldots, \bm{w}_K$ fixed.
\begin{prop}
\label{prop:ip1}
Suppose (C1) and (C2).
A matrix $W_z$ satisfies the stationary condition~\eqref{eq:KKT-Wz} if and only if
\begin{align}
\label{eq:ip1-Uz}
U_z &\leftarrow (W^{\htop} G_z )^{-1} E_z,
\\
\label{eq:ip1-Wz}
W_z &\leftarrow U_z \left(U_z^{\htop} G_z U_z  \right)^{-\frac{1}{2}} Q, \quad Q \in \mathrm{U}_{M - K},
\end{align}
where $\mathrm{U}_{M - K}$ is the set of all unitary matrices of size $M - K$,
and the last $M - K$ columns of $W$ in \eqref{eq:ip1-Uz}, $W_z$, are set to the current estimate before update.

Moreover, \eqref{eq:ip1-Uz}--\eqref{eq:ip1-Wz} with an arbitrary $Q \in \mathrm{U}_{M - K}$ globally minimizes $J_W$ if $\bm{w}_1, \ldots, \bm{w}_K$ are kept fixed.
\end{prop}
\begin{proof}
The first $K$ rows of \eqref{eq:KKT-Wz} are linear with respect to $W_z$ and solved as $W_z = U_z B$, where $U_z$ is defined by \eqref{eq:ip1-Uz} and $B \in \mathbb{C}^{(M - K) \times (M - K)}$ is a free parameter.
The remaining $M - K$ rows of \eqref{eq:KKT-Wz} constrain $B$ and we obtain \eqref{eq:ip1-Uz}--\eqref{eq:ip1-Wz},
which certainly satisfies the stationary condition~\eqref{eq:KKT-Wz}.

We next prove the latter statement.
As \eqref{eq:ip1-Uz}--\eqref{eq:ip1-Wz} satisfies \eqref{eq:KKT-Wz}, $\trace (W_z^{\htop} G_z W_z) = M - K$ holds.
It also holds that $| \det W | = | \det [W_s, W_z Q^{-1}] |$ for any $Q \in \mathrm{U}_{M - K}$.
Hence, $J_W$ takes the same value on all stationary points.
On the other hand, by Proposition~\ref{prop:opt-exist-sub}, $J_W$ attains its minimum at some $W_z$, and this $W_z$ must satisfy the stationary condition \eqref{eq:KKT-Wz}.
\end{proof}
From Proposition~\ref{prop:ip1}, as an update formula of $W_z$, we can use \eqref{eq:ip1-Uz}--\eqref{eq:ip1-Wz} with $Q = I_{M - K}$.
Although this formula guarantees the monotonic nonincrease of the cost function, the computation of \eqref{eq:ip1-Wz} is not efficient.
We therefore propose an acceleration of \eqref{eq:ip1-Uz}--\eqref{eq:ip1-Wz}, resulting in \eqref{eq:ip1-Wz-fast} below.

We now show that we can adopt \eqref{eq:ip1-Wz-fast} below as an update formula for $W_z$ if we do not need to separate the noise to satisfy \eqref{eq:z-gauss}.
Let $U_z$ and $W_z$ be defined by \eqref{eq:ip1-Uz} and \eqref{eq:ip1-Wz}, respectively.
Let also $W_z' \in \mathbb{C}^{M \times (M - K)}$ be a matrix satisfying $\image W_z' = \image U_z (= \image W_z)$, i.e., $W_z' = W_z R$ for some $R \in \mathbb{C}^{(M - K) \times (M - K)}$.
Then, concerning~\eqref{eq:ip0-u} and \eqref{eq:ip1-Uz}, we have
\begin{align}
\label{eq:ip0-u-equiv}
&(W'^{\htop} G_k)^{-1} \bm{e}_k = (W^{\htop} G_k)^{-1} \bm{e}_k, \quad k = 1, \ldots,K,
\\
\label{eq:ip1-Uz-equiv}
&\quad \image \left( (W'^{\htop} G_z)^{-1} E_z \right) = \image \left( (W^{\htop} G_z)^{-1} E_z \right),
\end{align}
where $W' \coloneqq [W_s, W_z']$.
The equality \eqref{eq:ip0-u-equiv} implies that using $W'$ instead of $W$ does not affect the resultant $\bm{w}_k$ obtained from \eqref{eq:ip0-u}--\eqref{eq:ip0-w} for each $k = 1, \ldots, K$.
Also, \eqref{eq:ip1-Uz-equiv} means that $\image U_z$ is updated to the same subspace in \eqref{eq:ip1-Uz} regardless of whether we use $W$ or $W'$ on the right hand side of \eqref{eq:ip1-Uz}.
Hence, it turns out that we only have to update $W_z$ so as to satisfy $\image W_z = \image U_z$, unless we care about the noise components.
Such update of $W_z$ can be attained simply by $W_z \leftarrow U_z \coloneqq (W^{\htop} G_z)^{-1} E_z$.
In this paper, inspired by~\cite{scheibler2019over-iva}, we propose a more efficient update:
\begin{align}
\label{eq:ip1-Wz-fast}
W_z &\leftarrow \begin{pmatrix}
-(W_s^{\htop} G_z E_s)^{-1} (W_s^{\htop} G_z E_z)
\\
I_{M - K}
\end{pmatrix} \in \mathbb{C}^{M \times (M - K)},
\end{align}
where $E_s \coloneqq [\, \bm{e}_{1}, \ldots, \bm{e}_{K} \,]$ and $W_s \coloneqq [\, \bm{w}_1, \ldots, \bm{w}_K \,] =W E_s$.
We can check using the block matrix inversion that $W_z$ obtained by \eqref{eq:ip1-Wz-fast} satisfies $\image W_z = \image U_z$, where $U_z$ is given by \eqref{eq:ip1-Uz}.
This concludes the derivation of IP-1 and the procedure of IP-1 is summarized in Algorithm~\ref{alg:ip1}.

\vspace{-0pt}
\subsubsection{The derivation of IP-2 only for $K = 1$}
\label{sec:ip-2}

When $K = M = 2$ (or $K = 1$ and $M = 2$), it is known that the problem \eqref{problem-W} can be solved directly through a generalized eigenvalue problem~\cite{degerine2006maxdet,ono2012iva-stereo,ono2018ip-ASJ}.
We here extend this direct method to the case where $K = 1$ and $M \geq 2$, which is summarized in Proposition~\ref{prop:ip2}.
\begin{prop}
\label{prop:ip2}
Suppose (C1) and (C2), and let $K = 1$ and $M \geq 2$.
Then, a matrix $W = [ \bm{w}_1, W_z ]$ satisfies the stationary condition \eqref{eq:KKT-w}--\eqref{eq:KKT-Wz} if and only if
\begin{align}
\label{eq:ip2-eig}
\bm{u}_1 &\in \mathbb{C}^M ~\text{satisfying} ~ G_z \bm{u}_1 = \lambda G_1 \bm{u}_1, 
\\
\label{eq:ip2-orth}
U_z &\in \mathbb{C}^{M \times (M - 1)} ~\text{satisfying} ~ U_z^{\htop} G_z \bm{u}_1 = \bm{0}_{M - 1},
\\
\label{eq:ip2-w1}
\bm{w}_1 &= \bm{u}_1 \left( \bm{u}_1^{\htop} G_1 \bm{u}_1  \right)^{- \frac{1}{2}} e^{\sqrt{-1}\theta}, \quad \theta \in \mathbb{R}, 
\\
\label{eq:ip2-Wz}
W_z &= U_z \left( U_z^{\htop} G_z U_z \right)^{- \frac{1}{2}} Q, \quad Q \in \mathrm{U}_{M - 1},
\end{align}
where \eqref{eq:ip2-eig} is the generalized eigenvalue problem with $\bm{u}_1$ and $\lambda \in \mathbb{R}$ being a generalized eigenvector and the corresponding eigenvalue.

Moreover, if $\lambda$ in \eqref{eq:ip2-eig} is chosen as the largest generalized eigenvalue, then any $W$ obtained by \eqref{eq:ip2-eig}--\eqref{eq:ip2-Wz} globally minimizes $J_W$.
\end{prop}
\begin{proof}
The ``if'' part is obvious and we prove the ``only if'' part.
The equations~\eqref{eq:KKT-w}--\eqref{eq:KKT-Wz} imply $\bm{w}_1^{\htop} G_1 \bm{w}_1 = 1$ and that $G_1 \bm{w}_1$ and $G_z \bm{w}_1$ are orthogonal to the subspace $\image W_z$ of dimension $M - 1$.
Hence, \eqref{eq:ip2-eig} and \eqref{eq:ip2-w1} are necessary.
Also, the equations~\eqref{eq:KKT-w}--\eqref{eq:KKT-Wz} together with \eqref{eq:ip2-w1} imply $W_z^{\htop} G_z W_z = I_{M - 1}$ and that $G_z \bm{u}_1$ are orthogonal to $\image W_z$.
Thus, \eqref{eq:ip2-orth} and \eqref{eq:ip2-Wz} are necessary.

We next prove the latter statement.
As \eqref{eq:ip2-eig}--\eqref{eq:ip2-Wz} satisfy \eqref{eq:KKT-w}--\eqref{eq:KKT-Wz}, the sum of the first and second terms of $J_W$ becomes $M$.
On the other hand, as for the $\log \det$ term, it holds that
\begin{align*}
| \det W |&= \left( \bm{u}_1^{\htop} G_1 \bm{u}_1 \right)^{-\frac{1}{2}} \cdot \det \left( U_z^{\htop} G_z U_z \right)^{-\frac{1}{2}} \cdot | \det U |
\\
&= \sqrt{\lambda} \det \left( U^{\htop} G_z U \right)^{-\frac{1}{2}} \cdot | \det U | = \sqrt{\lambda} \det (G_z)^{-\frac{1}{2}},
\end{align*}
where $U \coloneqq [\bm{u}_1, U_z]$ and we use $U_z^{\htop} G_z \bm{u}_1 = \bm{0}_{M - 1}$ in the second equality.
Hence, the largest $\lambda$ leads to the smallest $J_W$. 
\end{proof}

From Proposition~\ref{prop:ip2}, we can update $W = [\bm{w}_1, W_z]$ using \eqref{eq:ip2-eig}--\eqref{eq:ip2-Wz} with $\theta = 0$ and $Q = I_{M - 1}$, minimizing $J_W$ globally.
Note that updating $\bm{w}_1$ by \eqref{eq:ip2-eig} and \eqref{eq:ip2-w1} is independent of $W_z$ since $G_1$ and $G_z$ are independent of $W_z$.
We can thus simplify the procedure of IP-2, which is summarized in Algorithm~\ref{alg:ip2}.
Interestingly, because $G_z$ and $G_1$ can be viewed as the covariance matrices of the mixture and noise signals,
IP-2 turns out to be a MaxSNR beamformer~\cite{van2004optimum,warsitz2007maxsnr}.
In other words, OverIVA with IP-2 is a method that alternately updates the target-source power spectrum and MaxSNR beamformer.

\section{Relation to prior works}
\label{sec:relation}

The conventional OverIVA~\cite{scheibler2019over-iva}, denoted as OverIVA-OC, is an acceleration of IVA.
It exploits not only the independence of sources but also the orthogonality constraint (OC~\cite{koldovsky2018ive, koldovsky2017ive-eusipco, scheibler2019over-iva}):
 the sample correlation (or inner product) between the target sources and the noise is to be zero, i.e., $W_s^{\htop} G_z W_z = O$.
This constraint is nothing but the first $K$ rows in \eqref{eq:KKT-Wz},
meaning that the stationarity of the Gaussian noise implicitly implies OC.
In this subsection, we clarify that OverIVA-OC can also be obtained from the proposed OverIVA.

OverIVA-OC restricts $W_z$ to the form $W_z = \begin{pmatrix} -B_z \\ I_{M - K} \end{pmatrix}$,
which together with OC imply $B_z = -(W_s^{\htop} G_z E_s)^{-1} (W_s^{\htop} G_z E_z)$.
This relation between $B_z$ and $W_s$ gives an update rule of $W_z$.
The procedure of OverIVA-OC is summarized in Algorithm~\ref{alg:ip3}.
Note that the update of $W_z$ has to be done immediately after optimizing any other variables $\bm{w}_1, \ldots, \bm{w}_K$ so as to always enforce OC in the model.

It is easy to see that Algorithm \ref{alg:ip3} can also be obtained from the proposed OverIVA by selecting IP-3 as the optimization procedure of BCD.
Hence, OverIVA-OC can be viewed as a special case of our OverIVA.
The main advantage of our approach is that, by removing OC, we can develop several algorithms including IP-1 and IP-2 that are expected to be more efficient than IP-3.
In fact, the computational cost of IP-1 per iteration is slightly less than IP-3.
Also, the convergence speed of IP-2 is much faster than IP-3.
The advantages of IP-1 and IP-2 are shown experimentally (see Section~\ref{sec:exp}).

\begin{algorithm}[t]
\caption{OverIVA}
\label{alg:iva}
{\small
\begin{algorithmic}[1]
\STATE Set $\varepsilon_1, \varepsilon_2 \in \mathbb{R}_{\geq 0}$ (we set $\varepsilon_1 = 10^{-5}$ and $\varepsilon_2 = 10^{-1}$ in \S\ref{sec:exp}).
\STATE Initialization: $W(f) = I_M$ for all $f = 1,\ldots,F$.
\STATE $G_z(f) \leftarrow \frac{1}{T} \sum_{t = 1}^T \bm{x}(f,t) \bm{x}(f,t)^{\htop}$ for all $f$.
\REPEAT
\STATE $s_{k}(f,t) \leftarrow \bm{w}_k(f)^{\htop} \bm{x}(f,t)$ for all $k,f,t$.
\STATE $\lambda_k(t) \leftarrow \max \{ \frac{1}{F} \| \bm{s}_k(t) \|^2, \varepsilon_1 \}$ for all $k,t$.
\STATE $G_k(f) \leftarrow \frac{1}{T} \sum_{t = 1}^T \frac{\bm{x}(f,t) \bm{x}(f,t)^{\htop}}{\lambda_{k}(f,t)} + \varepsilon_2 I_M$ for all $k,f$.
\STATE Update $W(f)$ using IP-1, IP-2, or IP-3 for all $f$.
\STATE Normalization for numerical stability:
$c_k \coloneqq \frac{1}{T} \sum_t \lambda_k(t)$, $\lambda_k(t) \leftarrow \lambda_k(t) c_k^{-1}$ and $\bm{w}_k(f) \leftarrow \bm{w}_k(f) c_k^{-1/2}$ for all $k, f, t$.
\UNTIL{convergence}
\STATE In the case of IP-2, update $W_z(f)$ using \eqref{eq:ip1-Wz-fast} for all $f$.
\STATE The separation result is obtained by \eqref{eq:projection-back}.
\end{algorithmic}
}
\end{algorithm}
\begin{algorithm}[t]
\caption{IP-1}
\label{alg:ip1}
{\footnotesize
\begin{algorithmic}[1]
\FOR{$k = 1,\ldots,K$}
\STATE $\bm{u}_k(f) \leftarrow \left( W(f)^{\htop} G_k(f) \right)^{-1} \bm{e}_k$
\STATE $\bm{w}_k(f) \leftarrow \bm{u}_k(f) \left( \bm{u}_k(f)^{\htop} G_k(f) \bm{u}_k(f) \right)^{-\frac{1}{2}}$
\ENDFOR
\STATE $W_z(f) \leftarrow \begin{pmatrix}
- (W_s(f)^{\htop} G_z(f) E_s)^{-1} (W_s(f)^{\htop} G_z(f) E_z) \\
I_{M - K}
\end{pmatrix}$
\end{algorithmic}
}
\end{algorithm}
\begin{algorithm}[t]
\caption{IP-2 (only for $K = 1$)}
\label{alg:ip2}
{\footnotesize
\begin{algorithmic}[1]
\STATE Solve the generalized eigenvalue problem $G_z(f) \bm{u} = \lambda G_1(f) \bm{u}$ to obtain the eigenvector $\bm{u}$ corresponding to the largest eigenvalue.
\STATE $\bm{w}_1(f) \leftarrow \bm{u} \left( \bm{u}^{\htop} G_1(f) \bm{u}  \right)^{- \frac{1}{2}}$
\end{algorithmic}
}
\end{algorithm}
\begin{algorithm}[h]
\caption{IP-3 (The conventional OverIVA with OC~\cite{scheibler2019over-iva} is equivalent to the proposed OverIVA with IP-3 as shown in \S \ref{sec:relation})}
\label{alg:ip3}
{\footnotesize
\begin{algorithmic}[1]
\FOR{$k = 1,\ldots,K$}
\STATE Execute lines 2, 3, and 5 in Algorithm~\ref{alg:ip1} to update $\bm{w}_k(f)$ and $W_z(f)$.
\ENDFOR
\end{algorithmic}
}
\end{algorithm}

\section{Experiments}
\label{sec:exp}

We carried out experiments to compare the separation and runtime performances of the following four methods:
\begin{description}
\setlength{\parskip}{0pt}
\setlength{\itemsep}{0pt}
\item[AuxIVA~\cite{ono2011iva}:] The conventional auxiliary-function-based IVA~\cite{ono2011iva}, followed by picking the $K$ signals with largest powers.
\item[OverIVA-OC~\cite{scheibler2019over-iva}:] It is OverIVA(IP-3). See Algorithms~\ref{alg:iva} and \ref{alg:ip3}.
\item[OverIVA(IP-1) and OverIVA(IP-2):] See Algorithms~\ref{alg:iva}, \ref{alg:ip1}, and \ref{alg:ip2}.
\end{description}
As evaluation data, we generated synthesized convolutive mixtures of target speech and untarget white-noise signals.
To this end, we used point-source speech signals from SiSEC2008~\cite{vincent2009sisec2008} and selected a set of room impulse responses (RIR) recorded in the room E2A from RWCP Sound Scene Database~\cite{rwcp}.
Then, for a given numbers of speakers $K$, white noises $L$, and microphones $M$,
(i) we randomly picked $K$ speech signals and $K + L$ RIRs,
(ii) generated $L$ white noises,
(iii) convolved these $K + L$ point-source signals using the RIRs,
and (iv) added the obtained $K + L$ spatial images so that
$\mathrm{SINR} \coloneqq 10 \log_{10} \frac{\frac{1}{K} \sum_{k = 1}^K \sigma_k^2 }{\sum_{l = 1}^L \sigma_l^2 }$ [dB] becomes a specified value,
where $\sigma_k^2$ and $\sigma_l^2$ denote the variances of speech and white noise signals, respectively.
We generated 10 mixtures for each condition.

We initialized $W(f) = I_M$, and set the number of optimization iterations to 50 except set it to three in OverIVA(IP-2).
The sampling rate was 16 kHz, the reverberation time was 300 ms, the frame length was 4096 (256 ms), and the frame shift was $1/4$ of the frame length.

Table~\ref{table:exp-res} shows the BSS performance averaged over 10 samples in terms of SDR~\cite{vincent2006sdr} and real time factor (RTF) calculated as \textit{``the total computation time [sec] divided by the signal length (10 sec).''} 
As expected, the proposed OverIVA is faster than the other methods while providing the comparable SDR.
In particular, OverIVA(IP-2) specialized for $K = 1$ significantly outperforms the other methods.

\vspace{-3 mm}
\begin{table}[h]
\begin{center}
\caption{The resultant SDR  [dB] and real time factor (RTF)}
{\scriptsize
\label{table:exp-res}
\vspace{-0 mm}
\begin{tabular}{@{ }c@{ }c@{}|ccc|cccc@{ }} \hline
\multicolumn{2}{c|}{\#speeches \& \#noises} & \multicolumn{3}{c|}{$K = 1, L = 5$} & \multicolumn{4}{c}{$K = 2, L = 1$} \\
\multicolumn{2}{c|}{$\mathrm{SINR}$ [dB]} & \multicolumn{3}{c|}{0 dB} & \multicolumn{4}{c}{10 dB} \\
\multicolumn{2}{c|}{\#channels $(M)$} & 3 & 5 & 7 & 3 & 4 & 5 & 6 \\ \hline
Mixture & SDR & 0.0 & 0.0 & 0.0 & -0.4 & -0.4 & -0.5 & -0.4 \\
\hline
\multirow{2}{*}{AuxIVA~\cite{ono2010auxica}} 
& SDR & 3.9 & 3.5 & 4.1 & 5.7 & 7.2 & 7.6 & 7.3 \\
& RTF & 0.26 & 0.80 & 3.23 & 0.25 & 0.48 & 0.80 & 1.84 \\ 
\hline
\multirow{2}{*}{OverIVA-OC~\cite{scheibler2019over-iva}}
& SDR &  4.5 & 5.6 & 6.6 & 6.3 & 7.6 & 5.8 & 6.1 \\
& RTF & 0.10 & 0.22 & 0.73 & 0.19 & 0.27 & 0.41 & 0.81 \\
\hline
\multirow{2}{*}{OverIVA(IP-1)}
& SDR & 4.5 & 5.6 & 6.6 & 6.1 & 7.5 & 6.0 & 6.2 \\
& RTF & 0.10 & 0.22 & 0.73 & 0.18 & 0.25 & 0.38 & 0.76 \\
\hline
\multirow{2}{*}{OverIVA(IP-2)}
& SDR & 5.3 & 7.0 & 8.6 & - & - & - & - \\ 
& RTF & 0.017 & 0.046 & 0.10 & - & - & - & - \\ \hline
\end{tabular}
Note that OverIVA-OC and OverIVA(IP-1) are the same when $K = 1$.
\vspace{-3pt}
All algorithms were implemented in Python 3.7 and run on a laptop PC with 2.6 GHz Intel Core i7.
}
\end{center}
\end{table}
\vspace{-8 mm}

\section{Conclusion}
\vspace{-1 mm}
We proposed a computationally efficient IVA for overdetermined BSS, called OverIVA.
Unlike the previous OverIVA~\cite{scheibler2019over-iva} relying on the orthogonality constraint, our approach exploits only the independence of sources and the stationarity of the Gaussian noise.
We verified the effectiveness of the proposed OverIVA in the experiment.

\vspace{-1 mm}
\section{Appendix}
\vspace{-1 mm}
To give a proof of Proposition~\ref{prop:ip1}, we need Proposition~\ref{prop:opt-exist-sub} below, which is a  modification of Proposition~\ref{prop:opt-exist} provided in~\cite{degerine2004ica-ML,degerine2006maxdet,yeredor2009HEAD,yeredor2012SeDJoCo}.
The proofs of Propositions~\ref{prop:opt-exist} and~\ref{prop:opt-exist-sub} are almost the same and so we omit it.
\begin{prop}[see, e.g., \cite{degerine2004ica-ML,degerine2006maxdet,yeredor2009HEAD,yeredor2012SeDJoCo}]
\label{prop:opt-exist}
Suppose (C1).
Then, $J_W$ is lower bounded and attains its minimum.
\end{prop}
\begin{prop}
\label{prop:opt-exist-sub}
Suppose (C1) and let $W \coloneqq [W_s, W_z] \in \mathbb{C}^{M \times M}$ be a matrix.
If the submatrix $W_s$ is full column rank and $W_z$ is a variable,
then the function $J(W_z) \coloneqq \trace \left( W_z^{\htop} G_z W_z  \right) - 2 \log | \det W |$ is lower bounded and attains its minimum.
\end{prop}

\vfill\pagebreak
\newpage
\bibliographystyle{IEEEbib}
\bibliography{refs}

\begin{thebibliography}{10}

\bibitem{comon1994ica}
P. Comon,
\newblock ``Independent component analysis, a new concept?,''
\newblock {\em Signal processing}, vol. 36, no. 3, pp. 287--314, 1994.

\bibitem{cardoso1998bss}
J.-F. Cardoso,
\newblock ``Blind signal separation: statistical principles,''
\newblock {\em Proceedings of the IEEE}, vol. 86, no. 10, pp. 2009--2025, 1998.

\bibitem{kim2007iva}
T. Kim, H.~T. Attias, S.-Y. Lee, and T.-W. Lee,
\newblock ``Blind source separation exploiting higher-order frequency
  dependencies,''
\newblock {\em IEEE Transactions on Audio, Speech, and Language Processing},
  vol. 15, no. 1, pp. 70--79, 2007.

\bibitem{hiroe2006iva}
A. Hiroe,
\newblock ``Solution of permutation problem in frequency domain {ICA}, using
  multivariate probability density functions,''
\newblock in {\em Proc. ICA}, 2006, pp. 601--608.

\bibitem{scheibler2019over-iva}
R. Scheibler and N. Ono,
\newblock ``Independent vector analysis with more microphones than sources,''
\newblock in {\em Proc. WASPAA}, 2019.

\bibitem{friedman1974projection}
J.~H. Friedman and J.~W. Tukey,
\newblock ``A projection pursuit algorithm for exploratory data analysis,''
\newblock {\em IEEE Transactions on Computers}, vol. 100, no. 9, pp. 881--890,
  1974.

\bibitem{huber1985projection}
P.~J. Huber,
\newblock ``Projection pursuit,''
\newblock {\em The annals of Statistics}, pp. 435--475, 1985.

\bibitem{hyvarinen1997fastica}
A. Hyv{\"a}rinen and E. Oja,
\newblock ``A fast fixed-point algorithm for independent component analysis,''
\newblock {\em Neural computation}, vol. 9, no. 7, pp. 1483--1492, 1997.

\bibitem{hyvarinen1999fastica}
A. Hyvarinen,
\newblock ``Fast and robust fixed-point algorithms for independent component
  analysis,''
\newblock {\em IEEE transactions on Neural Networks}, vol. 10, no. 3, pp.
  626--634, 1999.

\bibitem{wei2015fastica}
T. Wei,
\newblock ``A convergence and asymptotic analysis of the generalized symmetric
  {FastICA} algorithm,''
\newblock {\em IEEE Transactions on Signal Processing}, vol. 63, no. 24, pp.
  6445--6458, 2015.

\bibitem{cruces2004blind}
S.~A. Cruces-Alvarez, A. Cichocki, and S. Amari,
\newblock ``From blind signal extraction to blind instantaneous signal
  separation: criteria, algorithms, and stability,''
\newblock {\em IEEE Transactions on Neural Networks}, vol. 15, no. 4, pp.
  859--873, 2004.

\bibitem{ono2010auxica}
N. Ono and S. Miyabe,
\newblock ``Auxiliary-function-based independent component analysis for
  super-{Gaussian} sources,''
\newblock in {\em Proc. LVA/ICA}, 2010, pp. 165--172.

\bibitem{ono2011iva}
N. Ono,
\newblock ``Stable and fast update rules for independent vector analysis based
  on auxiliary function technique,''
\newblock in {\em Proc. WASPAA}, 2011, pp. 189--192.

\bibitem{ono2012iva-stereo}
N. Ono,
\newblock ``Fast stereo independent vector analysis and its implementation on
  mobile phone,''
\newblock in {\em Proc. IWAENC}, 2012, pp. 1--4.

\bibitem{koldovsky2017ive-eusipco}
Z. Koldovsk{\`y}, P. Tichavsk{\`y}, and V. Kautsk{\`y},
\newblock ``Orthogonally constrained independent component extraction: Blind
  mpdr beamforming,''
\newblock in {\em Proc. EUSIPCO}, 2017, pp. 1155--1159.

\bibitem{koldovsky2018ive}
Z. Koldovsk{\`y} and P. Tichavsk{\`y},
\newblock ``Gradient algorithms for complex non-gaussian independent
  component/vector extraction, question of convergence,''
\newblock {\em IEEE Transactions on Signal Processing}, vol. 67, no. 4, pp.
  1050--1064, 2018.

\bibitem{murata2001projection-back}
N. Murata, S. Ikeda, and A. Ziehe,
\newblock ``An approach to blind source separation based on temporal structure
  of speech signals,''
\newblock {\em Neurocomputing}, vol. 41, no. 1-4, pp. 1--24, 2001.

\bibitem{pham1996ica-ml}
D.~T. Pham,
\newblock ``Blind separation of instantaneous mixture of sources via an
  independent component analysis,''
\newblock {\em IEEE Transactions on Signal Processing}, vol. 44, no. 11, pp.
  2768--2779, 1996.

\bibitem{pham2001ica-nonstationary}
D.-T. Pham and J.-F. Cardoso,
\newblock ``Blind separation of instantaneous mixtures of nonstationary
  sources,''
\newblock {\em IEEE Transactions on Signal Processing}, vol. 49, no. 9, pp.
  1837--1848, 2001.

\bibitem{yeredor2012SeDJoCo}
A. Yeredor, B. Song, F. Roemer, and M. Haardt,
\newblock ``A “sequentially drilled” joint congruence {(SeDJoCo)}
  transformation with applications in blind source separation and multiuser
  {MIMO} systems,''
\newblock {\em IEEE Transactions on Signal Processing}, vol. 60, no. 6, pp.
  2744--2757, 2012.

\bibitem{yeredor2009HEAD}
A. Yeredor,
\newblock ``On hybrid exact-approximate joint diagonalization,''
\newblock in {\em Proc. CAMSAP}, 2009, pp. 312--315.

\bibitem{degerine2004ica-ML}
S. D{\'e}gerine and A. Za{\"\i}di,
\newblock ``Separation of an instantaneous mixture of {Gaussian} autoregressive
  sources by the exact maximum likelihood approach,''
\newblock {\em IEEE Transactions on Signal Processing}, vol. 52, no. 6, pp.
  1499--1512, 2004.

\bibitem{degerine2006maxdet}
S. D{\'e}gerine and A. Za{\"\i}di,
\newblock ``Determinant maximization of a nonsymmetric matrix with quadratic
  constraints,''
\newblock {\em SIAM Journal on Optimization}, vol. 17, no. 4, pp. 997--1014,
  2006.

\bibitem{amari1996natural-gradient}
S. Amari, A. Cichocki, and H.~H. Yang,
\newblock ``A new learning algorithm for blind signal separation,''
\newblock in {\em Proc. NIPS}, 1996, pp. 757--763.

\bibitem{ono2018ip-ASJ}
N. Ono,
\newblock ``Fast algorithm for independent component/vector/low-rank matrix
  analysis with three or more sources,''
\newblock in {\em Proc. ASJ Spring Meeting}, 2018, (in Japanese).

\bibitem{ike2019ilrma-ft}
R. Ikeshita, N. Ito, T. Nakatani, and H. Sawada,
\newblock ``Independent low-rank matrix analysis with decorrelation learning,''
\newblock in {\em Proc. WASPAA}, 2019.

\bibitem{van2004optimum}
H.~L. Van~Trees,
\newblock {\em Optimum array processing: Part IV of detection, estimation, and
  modulation theory},
\newblock John Wiley \& Sons, 2004.

\bibitem{warsitz2007maxsnr}
E. Warsitz and R. Haeb-Umbach,
\newblock ``Blind acoustic beamforming based on generalized eigenvalue
  decomposition,''
\newblock {\em IEEE Transactions on Audio, Speech, and Language Processing},
  vol. 15, no. 5, pp. 1529--1539, 2007.

\bibitem{vincent2009sisec2008}
E. Vincent, S. Araki, and P. Bofill,
\newblock ``The 2008 signal separation evaluation campaign: A community-based
  approach to large-scale evaluation,''
\newblock in {\em Proc. ICA}, 2009, pp. 734--741.

\bibitem{rwcp}
S. Nakamura, K. Hiyane, F. Asano, T. Nishiura, and T. Yamada,
\newblock ``Acoustical sound database in real environments for sound scene
  understanding and hands-free speech recognition,''
\newblock in {\em LREC}, 2000.

\bibitem{vincent2006sdr}
E. Vincent, R. Gribonval, and C. F{\'e}votte,
\newblock ``Performance measurement in blind audio source separation,''
\newblock {\em IEEE Transactions on Audio, Speech, and Language Processing},
  vol. 14, no. 4, pp. 1462--1469, 2006.

\end{thebibliography}
\end{document}